\documentclass[aps,pra,
nofootinbib,longbibliography,notitlepage,
 11pt,
floatfix]{revtex4-1}

\usepackage{amsmath,amsthm,amsfonts,amssymb}
\usepackage{graphicx}
\usepackage{microtype}
\usepackage{braket}

\usepackage{xypic}

\usepackage[usenames,dvipsnames]{color}
\usepackage[colorlinks=true,citecolor=blue,linkcolor=magenta]{hyperref}
\usepackage{placeins}

\numberwithin{equation}{section}
\numberwithin{figure}{section}
\numberwithin{table}{section}

\newtheorem{theorem}{Theorem}[section]
\newtheorem{lemma}[theorem]{Lemma}

\theoremstyle{definition}
\newtheorem{definition}[theorem]{Definition}

\DeclareMathAlphabet{\mathpzc}{OT1}{pzc}{m}{it}

\newcommand{\FF}{\mathbb{F}}
\newcommand{\ZZ}{\mathbb{Z}}
\newcommand{\calS}{\mathcal{S}}
\newcommand{\calG}{\mathcal{G}}
\newcommand{\calQ}{\mathcal{Q}} % quantum code.

\newcommand{\fa}{\mathfrak{a}}
\newcommand{\fb}{\mathfrak{b}}
\newcommand{\fc}{\mathfrak{c}}

\newcommand{\norm}[1]{\left\| {#1}\right\|}

\newcommand{\ncheck}{{n_\text{\rm c}}}
\newcommand{\nout}{{n_\text{\rm out}}}
\newcommand{\nin}{n_\text{\rm inner}}

\newcommand{\kin}{k_\text{\rm inner}}

\newcommand{\nT}{n_T}
\newcommand{\CF}{\mathfrak{C}}

\begin{document}

\title{Towers of generalized divisible quantum codes}
\author{Jeongwan Haah}\email{jwhaah@microsoft.com}
\affiliation{Quantum Architectures and Computation, Microsoft Research, Redmond, Washington, USA}

% \date{\today}

\begin{abstract}
A divisible binary classical code is one in which 
every code word has weight divisible by a fixed integer.
If the divisor is $2^\nu$ for a positive integer $\nu$,
then one can construct a Calderbank-Shor-Steane (CSS) code,
where $X$-stabilizer space is the divisible classical code,
that admits a transversal gate in the $\nu$-th level of Clifford hierarchy.
We consider a generalization of the divisibility 
by allowing a coefficient vector of odd integers
with which every code word has zero dot product modulo the divisor.
In this generalized sense,
we construct a CSS code with divisor $2^{\nu+1}$ and code distance $d$
from any CSS code of code distance $d$ and divisor $2^\nu$
where the transversal $X$ is a nontrivial logical operator.
The encoding rate of the new code 
is approximately $d$ times smaller than that of the old code.
In particular, for large $d$ and $\nu \ge 2$,
our construction yields a CSS code of parameters $[[O(d^{\nu-1}), \Omega(d),d]]$
admitting a transversal gate at the $\nu$-th level of Clifford hierarchy.
For our construction we introduce a conversion from 
magic state distillation protocols based on Clifford measurements 
to those based on codes with transversal $T$-gates.
Our tower contains, as a subclass, generalized triply even CSS codes
that have appeared in so-called gauge fixing or code switching methods.
\end{abstract}

\maketitle

\section{Introduction}

The role of quantum error correcting codes is to protect states of qubits
and operate on them reliably. 
The most straightforward way of implementing fault-tolerant operations on encoded states
is to apply a tensor product operator on a code to induce desired action on the encoded state.
However, any nontrivial error correcting code cannot admit a universal set of operations this way~\cite{EastinKnill2009Restrictions},
and thus one has to inject magic states~\cite{Knill2004a,BravyiKitaev2005Magic},
or modify the code space dynamically%
~\cite{PaetznickReichardt2013Universal,
Bombin2015,
BravyiCross2015,
JochymOConnorBartlett2016,
Jones2016},
or seek for some other method to implement fault-tolerant operations~\cite{KnillLaflammeZurek1996Threshold,Yoder2016}.
Codes that admits Clifford operations are well understood, and concrete examples exist~\cite{Steane_1996,Bombin_2006},
but those that admit non-Clifford gates require nonlinear constraints,
making analysis and construction more difficult. 
For instance, a CSS code with transversal $T$ gate gives rise to a set of cubic equations 
on the stabilizer space~\cite{BravyiHaah2012Magic}.

Specifically, a CSS code of which any $X$-stabilizer has weight multiple of $2^\nu$
admits a transversal logical gate $T_\nu = \mathrm{diag}(1,\exp(2\pi i/2^\nu))$ 
in the $\nu$-level of Clifford hierarchy%
\footnote{
	The Clifford hierarchy~\cite{ChuangGottesman1999}
	is defined inductively by setting $\mathcal P_1$ to be the set of all Pauli operators on $n$ qubits,
	and $\mathcal P_\nu = \{ U \in \mathbb U( 2^n ) ~:~ U \mathcal P_1 U^\dagger \subseteq \mathcal P_{\nu -1} \} $
	for $\nu \ge 2$. The second level $\mathcal P_2$ is the Clifford group.
} 
$\mathcal P_\nu$~\cite{BravyiHaah2012Magic,LandahlCesare2013Tgate}
(the induced action of the transversal gate on the encoded qubit
may depend on a particular choice of the logical qubits, 
especially if there are many logical qubits in the code).
We call such a CSS code {\it level-$\nu$  divisible}~\cite{Ward1990}.

The purpose of this paper is to give a recipe for constructing CSS~\cite{CalderbankShor1996Good,Steane_1996}
codes that admit a transversal logical gate 
at the third and higher levels $\mathcal P_\nu$ of Clifford hierarchy.
We present a method to combine codes with a transversal gate (in fact, a diagonal phase gate)
at the $\nu$-th level, to make a code with a transversal gate at the $(\nu+1)$-th level.
Our level-lifting is a result of close examination on magic state distillation protocols.

There has been a number of interesting coincidences observed among distillation protocols.
It was pointed out~\cite{Reichardt2005} that
Knill's magic distillation scheme using the Steane code $[[7,1,3]]$
has the same output error probability as the Bravyi and Kitaev's scheme using the punctured Reed-Muller code 
$[[15,1,3]]$~\cite{Knill2004a,BravyiKitaev2005Magic}.
It was also observed~\cite{Jones2012} 
that Jones' check-based magic state distillation protocol using the $[[2k+4,2k,2]]$ code
has the same output error probability at the leading order
as Bravyi and Haah's protocol using the $[[6k+8,2k,2]]$ code~\cite{Jones2012,BravyiHaah2012Magic}.
More coincidences were found~\cite[App.~C]{HHPW2017magic}
between apparently distinct magic state distillation protocols,
but the pattern was obvious:
A distillation protocol based on Hadamard measurement routines 
is related to another protocol based on transversal $T=T_3$ gates in the third level $\mathcal P_3$.
Our construction shows that these connections are natural,
once the overall syndrome measurements are expressed as a single binary matrix.

The correspondence between the two classes of distillation protocols
constitutes the bottom of a tower of codes,
and provides an inductive step in the construction of higher levels of the tower.
Roughly speaking, we show that
a code with transversal $T_\nu$ can be used to distill $T_{\nu +1}$,
and a systematic representation of the syndrome measurements of the distillation protocol
gives a code with transversal $T_{\nu+1}$.
Although the construction was motivated by comparing two distillation protocols,
the resulting code does not have to bear an interpretation as a distillation protocol.
It is legitimate to think of our construction as a mapping from a binary matrix with certain divisibility
to a larger binary matrix with a higher level of divisibility.

We will begin by reviewing distillation protocols based on Hadamard measurements in Section~\ref{sec:review}.
The goal here is to give an overview of \cite{HHPW2017magic}.
In Section~\ref{sec:completeCheckMatrix} we find a succinct way 
to express all information of a Hadamard-measurement based distillation protocol
in a single binary matrix.
This can be used to calculate various probabilities for the distillation protocol.
In Section~\ref{sec:orthobases}
we study divisibility and a strong notion of orthogonality on a binary vector space.
This is a variant of \cite{Ward1990,BravyiHaah2012Magic}.
Section~\ref{sec:levellifting} contains the most important calculation of this paper,
establishing the divisibility of the matrix derived in Section~\ref{sec:completeCheckMatrix}.
In Section~\ref{sec:tgates_codes} we define quantum codes and examine transversal gates.
In Section~\ref{sec:examples}, we apply our construction to produce examples of codes,
whose parameters we previously did not know.
A consequence of our construction is that 
the triorthogonal codes of parameters $[[6k+8,2k,2]]$ in \cite{BravyiHaah2012Magic}
does not require a Clifford correction to induce logical $T$ gate.
We conclude with further discussion in Section~\ref{sec:discussion}.

We denote by $\vec 1$ a binary vector whose entries are all 1.
The length (dimension) of $\vec 1$ should be inferred from the context.

\section{Clifford measurements to Check matrices}

\subsection{Review of distillation protocols by Clifford measurements}
\label{sec:review}

One of the earliest ideas of magic state distillation~\cite{Knill2004a}
introduces a measurement in the basis of magic states.
Since the observable that is diagonal in the basis of magic states
cannot be a Pauli operator,
such a measurement routine necessarily includes non-Clifford gates realized by some magic state.
This may sound like circular nonsense,
but the point is that the measurement routine can be implemented inside an error detecting (or correcting) code,
so the non-Clifford gates in the measurement routine can be noisy.
The simplest such protocol is based on the Steane code $[[7,1,3]]$~\cite{Steane_1996}.
The Steane code has a property that the transversal Hadamard gate $H^{\otimes 7}$ induces
a Hadamard gate on the logical qubit.
If we implement controlled Hadamards on the physical qubits approximately,
then the induced controlled Hadamard on the logical qubit will have higher fidelity,
and the post-measurement state, which is a magic state,
will also be of higher fidelity.

This idea has led to many interesting protocols such as 
one that operates on 5 qubits including measurement ancilla
to distill $T$ states~\cite{MeierEastinKnill2012Magic-state},
one that can measure small Fourier angles~\cite{DP2014},
and one that uses a very small number of noisy $T$-gates 
per output in a very high fidelity regime~\cite{Jones2012}.
All these protocols have a common feature 
that they use a special error-detecting code that admits transversal Clifford gates.
Recently, we have introduced general criteria and several constructions for magic state distillation protocols 
based on the Clliford measurements~\cite{HHPW2017magic,enum,rtrio},
which we review here.
For simplicity, let us explain protocols only for the $+1$ eigenstate of Hadamard $H$.
As usual, we may assume that the error model is a dephasing channel in the magic basis,
i.e., an independent noisy magic state suffers from a $Y$ error with probability $p$~\cite{BravyiKitaev2005Magic}.

We separate the problem into two codes, namely inner and outer codes.
Let us first explain inner codes.
An {\bf inner code} $[[\nin,\kin,d]]$
is a weakly self-dual CSS code on $n$ qubits, i.e., $H^{\otimes \nin}$ is a logical operator,
where
we find a decomposition of the code space (a logical operator basis)
such that the transversal Hadamard induces either the product of logical Hadamards over all logical qubits,
or the product of logical swaps.
The two possibilities are exhaustive and mutually exclusive:
If the transversal Hadamard becomes logical Hadamard for a suitable choice of the logical operators,
which is the case when the stabilizer group does {\it not} contain $X^{\otimes \nin}$, the code is called {\bf normal};
otherwise, it is called {\bf hyperbolic}.
For example, 
the Steane code $[[7,1,3]]$ is normal, and a 2D color code~\cite{Bombin_2006} encoding one qubit is normal,
but the 4-qubit code with stabilizers $XXXX,ZZZZ$ is hyperbolic.
One can obtain a normal code from a hyperbolic code by puncturing one qubit~\cite{HHPW2017magic}.

To emphasize a particular choice of logical operators for a normal weakly self-dual CSS code,
we will speak of a {\bf normal basis} of logical operators.
In a normal basis, each pair of $X$- and $Z$-logical operators for a logical qubit $a = 1, \ldots, \kin$ 
is associated with a bit string $\ell^{(a)}$ such that $\ell^{(a)} \cdot \ell^{(b)} = \delta_{ab}$.
Whenever the stabilizer group does not contain $X^{\otimes \nin}$,
one can always find a normal basis of logical operators by a Gram-Schmidt process.
Below and throughout this paper we only consider normal codes with some normal basis.

With a normal inner code, we can implement an {\bf $H$-measurement routine} 
to measure the eigenvalue of $H^{\otimes \kin}$.
The net effect of the measurement routine, 
post-selected on satisfactory Pauli stabilizer measurements (inner syndrome) on the normal inner code,
is the following trace-decreasing map $\mathcal D$ of density matrices.
\[
 \mathcal D(\rho) = \Pi_+ \rho \Pi_+ + O(p^2) \Pi_- \rho \Pi_- + O(p^d) \mathcal L(\rho )
\]
where $\rho$ is a logical state, $\Pi_\pm$ are the projectors onto the $\pm 1$ eigenspace of $H^{\otimes k}$, respectively,
$p$ is the independent error probability of individual $T_Y = \exp( i \pi Y / 8)$ gates%
\footnote{$T_Y$ and $T_3 = \mathrm{diag}(1,e^{i\pi/4})$ are conjugate to each other by a Clifford gate up to a global phase.}
suffering from $Y$ errors
which are used to implement $\mathcal D$,
and $\mathcal L$ is some transformation within the code space of the inner code (i.e., logical errors).
The second term (where the measurement outcome is wrong)
is because the measurement routine makes use of an identity 
\[
^C H = \underbrace{e^{-i\pi Y/8}}_{T_Y^\dagger} (^C Z) \underbrace{e^{i \pi Y/8}}_{T_Y}
\]
where $^CU = \ket 0 \bra 0 \otimes I + \ket 1 \bra 1 \otimes U$ for any unitary $U$,
and if the both $T_Y$-gates have $Y$ errors, then the measurement outcome will be flipped.
Note that we assumed that physical measurements are perfect;
the all possible errors are from $T_Y$-gates.

Given a map $\mathcal D$, we can check the parity of the eigenvalues of $H^{\otimes \kin}$ on two or more magic states,
and obtain a distillation protocol with quadratic error suppression.
If the quadratic error suppression is the goal, 
then it is a good strategy to choose an inner code of as high encoding rate as possible with $d=2$~\cite{Jones2012}.
If higher order error suppression (order $d>2$) is needed,
then we should consider further measurements.
For simplicity, let us restrict to the case where we only have one inner code.
Sending a fixed set of magic states through the map $\mathcal D$ multiple times,
does not improve the fidelity of the output magic states.
Instead, we consider a larger number of input magic states to check,
and specify which subset of the input magic states is given to the map $\mathcal D$.

This is the place where an {\bf outer code} becomes necessary.
Let us represent input magic states to be tested by bits: 
0 if the input magic state is the correct one,
and 1 if the input has $Y$ error on it.
Ideally, we would want the all-0 bit configuration,
but for the purpose of distillation protocol that allows $\ge d$ errors,
it suffices to filter out small weight $<d$ configurations.
Hence, each measurement routine $\mathcal D$ imposes parity checks on the bit configurations,
and the parity checks must define a classical code on the bit configurations of minimum distance $\ge d$.

The distance requirement is necessary but not sufficient.
This is due to the second term in $\mathcal D$;
we cannot always trust a parity check's outcome.
It is thus necessary to put some redundancy into the set of parity checks,
making the whole set of the checks {\bf sensitive}.
If we express the sets of bits that are tested by $H$-measurement routines as the rows of a parity check matrix $M$,
and the bit representation of the input magic states as a binary vector $e$,
then the sensitivity condition together with the minimum distance requirement is expressed as
\begin{align}
	|e| + 2|Me| \ge d \quad \text{ for any nonzero } e 
\label{eq:sensitivity}
\end{align}
where $|\cdot|$ denotes the Hamming weight.
Sometimes this means that $M$ should have linearly dependent rows.
As a simple example, consider input magic states that are arranged on a line of $\nout$ points,
and suppose we want quartic order error suppression using some normal inner code of $d \ge 4$.
Suppose we demand that the outer code be a repetition code of distance $\nout$.
The parity checks for this outer code can be nearest-neighbor parity checks,
but then an erroneous input magic state at one end of the line, which may exist with probability of order $p$,
will be tested only once, and there is overall $O(p^3)$ probability for the faulty magic state to pass the entire test.
Hence we must consider parity checks for nearest neighbors as if the line were a ring.
Then, every input magic state is tested twice, and a mistake from one $H$-measurement routine can be detected by another.

If we use a single inner code of distance $\ge d$,
then \eqref{eq:sensitivity} is a necessary and sufficient condition for a parity check matrix of an outer code
to yield a distillation protocol of error reduction order $\ge d$~\cite{HHPW2017magic}.
However, one can consider various inner codes and design a protocol where input magic states
are tested in a certain order. A full protocol is then defined by a sequence of inner codes 
$[[\nin^{(\alpha)},\kin^{(\alpha)},d^{(\alpha)}]]$,
and a sequence of sets of qubits $M^{(\alpha)}$ that is to be tested by $\alpha$-th inner code.
In this case, one should show that any pattern of errors, including those inside $H$-measurement routines,
that may pass the entire protocol has weight $\ge d$.
An example is a protocol where 
one checks a single input magic state by a sequence of inner codes of distance $d^{(\alpha)} = 2\alpha +1$
with $\alpha = 1,2,3,\ldots,\ncheck$.
This example outputs a magic states whose error rate is reduced to $d$-th order where $d = 2 \ncheck + 1$.

\subsection{Representation of errors and checks}
\label{sec:completeCheckMatrix}

We imagine that $H$-measurement routines are implemented sequentially 
to check $\nout$ magic states.
(They may operate on disjoint sets of qubits,
in which case they can be implemented
in parallel, but we can always regard them 
as a sequential application of $H$-measurements, one at a time.)
Here, we express the entire protocol as a collection of parity checks on all $T_Y$-gates/states
including those that are used to implement the $H$-measurement routines.
We consider normal codes only, and thus each $H$-measurement routine $\alpha$ consumes
$2\nin^{(\alpha)}$ $T_Y$ gates.

Under the stochastic error model,
any possible error is of $Y$-type and hence corresponds to a bit string.
Let $y^{(0)} = (y^{(0)}_1,\ldots, y^{(0)}_\nout )$
denote any error bit string on the magic states that are tested.
Let $y^{(\alpha)}$ for $\alpha = 1,2,\ldots,\ncheck$ denote
the error bit string inside the $\alpha$-th $H$-measurement routine.
Since there are two layers of $T_Y$-gates on the inner code,
$y^{(\alpha)}$ with $\alpha \ge 1$ must have even length;
for notational convenience, 
we let the first half of $y^{(a)}$ to be the error pattern
before the layer of $^CZ$, and the second half to be that after $^CZ$.
Thus, the error pattern in the entire protocol is represented by
a bit vector 
\begin{align}
y =
 \begin{pmatrix}
 	y^{(0)} & y^{(1)} & \cdots & y^{(\ncheck)}
 \end{pmatrix}
\label{eq:errorVector}
\end{align}
of length $\nT = \nout + 2\sum_{\alpha=1}^{\ncheck} \nin^{(\alpha)}$,
which is by definition equal to the total number of $T_Y$-gates/states.
We will regard the error bit vector as a column vector ($\nT$-by-$1$ matrix).

The protocol specifies for each $H$-measurement routine $\alpha$ a set of magic states to be tested.
This set $M^{(\alpha)}$ can be viewed as a bit string of length $\nout$;
$M^{(\alpha)}_i = 1$ if and only if the qubit $i$ is included in the test by the routine $\alpha$.
If an $H$-measurement routine $\alpha$ was perfect,
then it would read off the parity of the error in $M^{(\alpha)}$,
which is equal to $M^{(\alpha)} \cdot y^{(0)}$ over $\FF_2$.
As we will use $M^{(\alpha)}$ as a submatrix of a complete check matrix,
we regard $M^{(\alpha)}$ as a $1$-by-$\nout$ matrix.

To take the outer measurement error into account,
we regard that any error on the pre-$^CZ$ layer is moved to the post-$^CZ$ layer:
\begin{align}
 (^{C_0}Z_j) Y_j = Z_0 Y_j (^{C_0}Z_j) .
\end{align}
This convention allows us to consider only the pre-$^CZ$ errors in regards to the outer syndrome.
That is, if the error pattern on magic states under the test is $y^{(0)}$,
then the outer syndrome by an $H$-measurement routine $\alpha$ is given by
\begin{align}
 M^{(\alpha)} y^{(0)} + 
 \underbrace{\begin{pmatrix} 1 & 1 & \cdots & 1 & 0 & 0 \cdots & 0 \end{pmatrix}}%
 _{(10)\otimes \vec 1_{\nin^{(\alpha)}}}
 y^{(\alpha)} \mod 2
\end{align}

The error pattern $y^{(0)}$ is not necessarily the same 
as we proceed from one $H$-measurement routine to the next.
This is because there could be a logical error from the inner code
that propagates to the magic states under the test.
If a bit string 
\begin{align}
\ell^{(\alpha)} = L^{(\alpha)}(y^{(\alpha)})
\end{align}
denotes the error on the $\nout$ magic states,
resulting from the logical error of the $\alpha$-th routine,
then after the $\alpha$-th measurement routine the error pattern is
\begin{align}
 y^{(0)}(\alpha) = y^{(0)}(0) + \ell^{(1)} + \cdots + \ell^{(\alpha)} \mod 2
\label{eq:ErrorPropagation}
\end{align}
and the outer syndrome by $\alpha$-th routine is given by
\begin{align}
 &M^{(\alpha)} y^{(0)}(0) + M^{(\alpha)}\ell^{(1)} + \cdots + M^{(\alpha)} \ell^{(\alpha-1)} + ((10) \otimes \vec 1_{\nin^{(\alpha)}}) y^{(\alpha)} \mod 2.
 \label{eq:outerSyndrome}
\end{align}

We have not yet determined the function $L^{(\alpha)}$ 
that maps an error pattern inside the inner code to a logical error.
This function is well-defined only if the encoded state on the inner code
remains in the code space.
The condition for the zero inner syndrome 
is given by commutation relations between an error operator and the stabilizers.
Since the errors are of $Y$ type and the stabilizers are from a weakly self-dual CSS code,
the commutation relation can be read off from
\begin{align}
 S^{(\alpha)} ( y^{(\alpha)}_\text{first half} + y^{(\alpha)}_\text{second half} ) 
 = ((11)\otimes S^{(\alpha)}) y^{(\alpha)} \mod 2
 \label{eq:innerSyndrome}
\end{align}
where each row of $S^{(\alpha)}$ corresponds to a stabilizer generator.
The sum of two halves is
to account for the fact that a pair of errors on a single qubit in the inner code
is equal to no error on that qubit in regards to inner syndrome.
Conditioned on the zero inner syndrome,
the function $L^{(\alpha)}$ can now be determined.
An encoded qubit $k$ is acted on by a $Y$-logical operator
if and only if the logical $X$ and $Z$ operators of the logical qubit $k$
both anticommute with a given Pauli operator.
Since the $X$ and $Z$ logical operators are represented by the same bit string 
$(L_{k,1},L_{k,2},\ldots,L_{k,\nin^{(\alpha)}})$
under the normal basis of logical operators we chose,
the $Y$ logical operator on the logical qubit $k$ is enacted if and only if
\begin{align}
 &(L_{k,1},L_{k,2},\ldots,L_{k,\nin^{(\alpha)}}) \cdot 
 (y^{(\alpha)}_\text{first half} +y^{(\alpha)}_\text{second half}) = 1 \mod 2.
\end{align}
Therefore, the function $L^{(\alpha)}$ is linear over $\FF_2$
whose matrix has rows that correspond to logical operators.
Since a routine $\alpha$ may not act on all the $\nout$ qubits,
the matrix $L^{(\alpha)}$, which is $\nout$-by-$\nin^{(\alpha)}$,
has nonzero rows only for the support of $M^{(\alpha)}$.
In addition, the choice of logical operators according to a normal basis 
ensures that nonzero rows of $L^{(\alpha)}$ are orthonormal.
\begin{align}
 \sum_{j=1}^{\nin^{(\alpha)}} L^{(\alpha)}_{kj} L^{(\alpha)}_{k'j} &= M^{(\alpha)}_k \delta_{kk'} \mod 2. %\label{eq:LLM}
\end{align}

We thus have collected all necessary ingredients for a complete check matrix:
A representation of errors \eqref{eq:errorVector},
the outer syndrome \eqref{eq:outerSyndrome}, 
and the inner syndrome \eqref{eq:innerSyndrome}.
\begin{widetext}
\begin{align}
\CF_0 =\left[
\begin{array}{c|c|c|c|c|c}
M^{(1)} & (10)\otimes \vec 1           &                             &                     &          &                          \\
M^{(2)} & M^{(2)}((11)\otimes L^{(1)} )& (10)\otimes \vec 1           &                     &         &                          \\
M^{(3)} & M^{(3)}((11)\otimes L^{(1)} )& M^{(3)}((11)\otimes L^{(2)} )& (10) \otimes \vec 1  &        &                          \\
\vdots  &          \vdots              &          \vdots              &        \vdots        &  \ddots&                           \\
M^{(\ncheck)}&M^{(\ncheck)}((11)\otimes L^{(1)})& M^{(\ncheck)}((11)\otimes L^{(2)})& \cdots & \cdots & (10)\otimes \vec 1        \\
\hline
 & (11)\otimes S^{(1)} &  &  &  &  \\
 &                     &(11) \otimes S^{(2)} &   &  &  \\
 & &  &  \ddots & & \\
 &  & &  & (11) \otimes S^{(\ncheck-1)} & \\
 &  & &  & & (11) \otimes S^{(\ncheck)}
\end{array}
\right]
\end{align}
\end{widetext}
where the blank entries are zero, and the vertical lines are to facilitate reading.
In the outer syndrome block, each displayed row is a single row,
whereas in the inner syndrome block, each displayed entry is a submatrix.
The propagated error from the inner codes to the output magic states
is inscribed in \eqref{eq:ErrorPropagation},
which we can represent as a linear map
\begin{widetext}
\begin{align}
\CF_1 =\left[
\begin{array}{c|c|c|c|c|c}
 I_\nout & (11)\otimes L^{(1)}          & (11) \otimes L^{(2)}         & (11) \otimes L^{(3)} & \cdots & (11) \otimes L^{(\ncheck)}
 \end{array}
\right].
\end{align}
\end{widetext}
The vertical lines are aligned with those of $\CF_0$.

We have arrived at the following theorem.
\begin{theorem}
Suppose the error pattern of all $T_Y$-gates and states is represented by a bit vector $y$.
Then, the magic state distillation protocol based on Hadamard check routines using a normal weakly self-dual code
accepts the output if and only if $\CF_0 y = 0$,
and the accepted output does not contain an error if and only if $\CF_1 y = 0$.
\label{thm:completeCheckMatrix}
\end{theorem}

\section{Orthogonal bases at level $\nu$}
\label{sec:orthobases}

In this section, it is implicit that an element of $\FF_2$ is promoted to an integer in $\ZZ$.
The association is such that $\FF_2 \ni 0 \mapsto 0 \in \ZZ$ and $\FF_2 \ni 1 \mapsto 1 \in \ZZ$.
Likewise, any element of $\ZZ/2^\nu \ZZ$ will be represented by an integer from $0$ to $2^\nu-1$ of $\ZZ$.
Unless noted by ``mod $ 2^\nu$,'' all arithmetics are over the usual integers $\ZZ$.
However, every vector space is over $\FF_2$.
We always regard a matrix as a set of rows.
The $r$-th row of a matrix $A$ is denoted as $A_{r*}$.

\begin{definition}
We consider a vector space $\FF_2^n$ equipped with an odd integer vector $t \in (2\ZZ+1)^n$,
called a {\bf coefficient} vector.
Let $\nu \ge 2$.
The {\bf norm at level $\nu$} 
($\nu$-norm) of $v \in \FF_2^n$ is $\| v \|_\nu = \sum_i v_i t_i \mod 2^\nu$.
Two vectors $v$ and $w$ are {\bf orthogonal at level $\nu$} ($\nu$-orthogonal)
if $\sum_i v_i w_i t_i = 0 \mod 2^{\nu -1}$.
Two subspaces $V,W$ are $\nu$-orthogonal 
if every pair $v \in V$ and $w \in W$ is $\nu$-orthogonal.
A set $G = \{ g^{(1)}, \ldots, g^{(k)} \}$ of vectors is
$\nu$-orthogonal
if, for any two disjoint subsets $G'$ and $G''$ of $G$, the subspaces $\mathrm{span}(G')$ and $\mathrm{span}(G'')$ are $\nu$-orthogonal.
A $\nu$-orthogonal set is {\bf $\nu$-orthonormal} if every member has $\nu$-norm 1.
A $\nu$-orthogonal set or subspace is {\bf $\nu$-null} if every member has $\nu$-norm 0.
To emphasize the coefficient vector $t$, 
we will sometimes write $(\nu,t)$-orthogonal, -norm ($\norm{\cdot}_{\nu,t}$), -null, -orthonormal.
\end{definition}

The $\nu$-orthogonality is a generalization of a notion 
that is considered for transversal $T$-gates~\cite{Bombin2015,KubicaBeverland2015,BravyiCross2015};
previously, a coefficient vector had components $\pm 1$.
Being $2$-orthogonal is equivalent to being orthogonal under the usual dot product over $\FF_2$.
Being $2$-null or $2$-orthonormal is nontrivial,
but is easily satisfied as Lemma~\ref{lem:nu2basis} below shows.
As we will see below,
an orthogonal matrix at level 3 is triorthogonal~\cite{BravyiHaah2012Magic} 
since $t_i$ is odd, but we do not know whether every triorthogonal matrix is 
orthogonal at level 3 for some coefficient vector $t$.
To the best of our knowledge,
all examples of triorthogonal matrices in the literature
are actually orthogonal at level 3.
See further discussion in Section~\ref{sec:examples}.

We now give equivalent conditions for the $\nu$-orthogonality,
as an application of a result by Ward~\cite{Ward1990};
see also~\cite{BravyiHaah2012Magic}.
It will be very convenient to introduce the following notation.
It is customary to denote a matrix element as $A_{ai}$ for a matrix $A$.
We  write $A^\star_{\fa i}$ 
for any unordered set of row labels $\fa = \{\fa_1, \ldots, \fa_m\}$
(whose cardinality $|\fa|$ is equal to $m$),
to denote
\begin{align}
 A^\star_{\fa i} = A_{\fa_1,i} A_{\fa_2,i} \cdots A_{\fa_m,i}.
\end{align}
If $|\fa| = 0$, then $A^\star_{\fa i} = 1$ by convention.
By definition, $|\fa|$ cannot be larger than the number of rows of $A$.

\begin{lemma}
Let $t$ be a coefficient vector of length $n$.
Let $A$ be a binary matrix with $n$ columns where the rows are $\FF_2$-linearly independent.
The following are equivalent:
\begin{enumerate}
	\item[(i)] The rows of $A$ form a $(\nu,t)$-orthogonal set.
	\item[(ii)] For every subset $K$ of rows of $A$, 
	$\norm{ \sum_{r \in K} A_{r*} \mod 2 }_{\nu,t} =  \sum_{r \in K} \norm{A_{r*} }_{\nu,t} \mod 2^\nu$.
	\item[(iii)]$ 2^{|\fa| -1} \sum_i A^\star_{\fa i} t_i = 0 \mod 2^\nu$ for every $\fa$ such that $2 \le |\fa| \le \nu$.
\end{enumerate}
In particular, the following are equivalent:
\begin{enumerate}
	\item[(1)] Every vector in a subspace $\calS \subseteq \FF_2^n$ has zero $(\nu,t)$-norm.
	\item[(2)] There exists a spanning set for $\calS$ that is $(\nu,t)$-null.
	\item[(3)] Every spanning set for $\calS$ is $(\nu,t)$-null.
\end{enumerate}
\label{lem:nu-orthogonality}
\end{lemma}

As an example, if the rows of a binary matrix $A$ are $\nu$-null with respect to $t = \vec 1$,
then any vector in the row $\FF_2$-span of $A$ has weight divisible by $2^\nu$.
More concretely, the rows of
\begin{align}
\begin{bmatrix}
	1 & 0 & 1 & 0 & 1 & 0 & 1\\
	0 & 1 & 1 & 0 & 0 & 1 & 1\\
	0 & 0 & 0 & 1 & 1 & 1 & 1
\end{bmatrix}
\end{align}
are $2$-null with respect to $t=\vec 1$ and span a doubly even subspace of $\FF_2^7$.

\begin{proof}
If $y \in \FF_2^n$ is a binary vector of weight $|y|$, 
then the parity $\epsilon(y) = 0,1$ of its weight can be expressed as
\begin{align}
\epsilon(y) = \frac12 \left(1 - (1 - 2)^{|y|} \right) = \sum_{p=1}^{|y|} \binom{|y|}{p} (-2)^{p-1}.
\end{align}
Here, the first equality is obvious; consider the cases where $|y|$ is even or odd.
The second equality is by the binomial theorem $(1-2)^m= \sum_{p=0}^m \binom{m}{p} 1^{m-p} (-2)^{p}$.
We claim that
\begin{align}
\epsilon(y) = \sum_{\fa \neq \emptyset} (-2)^{|\fa|-1} y^{\star}_\fa
\label{eq:parityMod2nu}
\end{align}
where we have treated the vector $y$ as a column matrix with $n$ rows to use the notation $y^\star_\fa$,
and $\fa$ ranges over all nonempty subsets of $\{1,2,\ldots, n\}$.
Eq.~\eqref{eq:parityMod2nu} can be understood as follows.
We may express the summation $\sum_{\fa \neq \emptyset}$ as a double summation $\sum_{p=1}^n \sum_{\fa: |\fa|=p}$.
Consider $\sum_{\fa: |\fa|=p} y^\star_\fa$.
The binary summand $y^\star_\fa$ is 1 if and only if $\fa$ labels nonzero entries of $y$.
There are $|y|$ nonzero entries of $y$, and there are $\binom{|y|}{p}$ different subsets of $p$ nonzero entries.
Hence, for $p > 0$ we see $\sum_{\fa: |\fa|=p} y^\star_\fa = \binom{|y|}{p}$,
and Eq.~\eqref{eq:parityMod2nu} is proved.

Since Eq.~\eqref{eq:parityMod2nu} is true as an integer equation,
we can take it modulo $2^\nu$.
Hence, for any binary matrix $B$,
if a vector $g$ is the sum of all rows of $B$ over $\FF_2$,
then
\begin{align}
 g_i = \sum_{\fa \neq \emptyset} (-2)^{|\fa|-1} B^{\star}_{\fa i} \mod 2^\nu.
 \label{eq:gExpansion}
\end{align}
We will use Eq.~\eqref{eq:gExpansion} with various $B$.

\emph{(i)} $\Rightarrow$ \emph{(ii)}:
If $|K| = 0$ or $1$, there is nothing to prove.
We will use induction in $|K|$.
Suppose \emph{(ii)} is true for any $K$ with $|K| < k$,
and we wish to prove \emph{(ii)} for $K = K' \sqcup \{ a \}$ with $|K'| = k-1$.
Put $g = \sum_{r \in K'} A_{r*} \mod 2$, that is a binary vector.
By the induction hypothesis, $\norm{g}_\nu = \sum_{r \in K'} \norm{A_{r*}}_\nu$.
Then, by \eqref{eq:gExpansion} $\mod 2$,
\begin{align*}
\norm{ \sum_{r \in K} A_{r*} \mod 2}_\nu 
&=
\norm{ A_{a*} +  g \mod 2}_\nu \\
&= 
\norm{ A_{a*} + g - 2 A_{a*} g }_\nu \qquad \text{(where the product of vectors is component-wise)}\\
&= 
\norm{A_{a*}}_\nu + \norm{g}_\nu - 2 \sum_i A_{ai}g_i t_i \mod 2^\nu\\
&=
\norm{A_{a*}}_\nu + \sum_{r \in K'} \norm{A_{r*}}_\nu \mod 2^\nu
\end{align*}
where in the last equality we used the assumption \emph{(i)} to eliminate the third term.
This completes the induction.

\emph{(ii)} $\Rightarrow$ \emph{(iii)}:
By \eqref{eq:gExpansion}, we have
\begin{align*}
0
=\norm{\sum_{r \in \fa} A_{r*} \mod 2}_\nu - \sum_{r \in \fa} \norm{ A_{r*}}_\nu 
= \sum_{\fa' \subseteq \fa, ~|\fa'|\ge 2} \sum_i (-2)^{|\fa'|-1} A^\star_{\fa' i} t_i \mod 2^\nu.
\end{align*}
If $|\fa| = 2$, then it must be that $\fa' = \fa$,
and \emph{(iii)} is proved in this case.
For $|\fa| > 2$ we use induction.
The summand of $\sum_{\fa' \subseteq \fa,~|\fa'| \ge 2}$ with $|\fa'| < |\fa|$ vanishes by the induction hypothesis,
and we are left with the summand with $\fa' = \fa$ which must vanish.
This completes the induction.

\emph{(iii)} $\Rightarrow$ \emph{(i)}:
Let $K$ and $K'$ be two disjoint sets of rows of $A$.
We have to show that $v = \sum_{r \in K} A_{r*} \mod 2$ 
and $w = \sum_{r \in K'} A_{r*} \mod 2$ are $\nu$-orthogonal.
Using \eqref{eq:gExpansion},
\begin{align*}
(-2) \sum_i v_i w_i t_i 
&= 
\sum_{\emptyset \neq \fa \subseteq K} \sum_{\emptyset \neq \fb \subseteq K'} 
(-2)^{|\fa|+|\fb|-1} \sum_i A^\star_{\fa i} A^\star_{\fb i} t_i &\mod 2^\nu \nonumber \\
&=
\sum_{\emptyset \neq \fc \subseteq K \cup K'} (-2)^{|\fc|-1} \sum_i A^\star_{\fc i} t_i &\mod 2^\nu \nonumber \\
&=
0 &\mod 2^\nu
\end{align*}
where the second equality is because $\fa$ and $\fb$ are disjoint,
and the third is because $|\fc| \ge 2$.

These complete the proof of the first part of the lemma.

The second part of the lemma uses \emph{(i) $\leftrightarrow$ (ii)}.
\emph{(3)} $\Rightarrow$ \emph{(2)}: Trivial.
\emph{(2)} $\Rightarrow$ \emph{(1)}: A null spanning set is orthogonal,
so by \emph{(i) $\rightarrow$ (ii)} we see every vector in the span has norm zero.
\emph{(1)} $\Rightarrow$ \emph{(3)}: 
Choose any spanning set $A$ for $\calS$. 
We only have to show that $A$ is orthogonal,
but this follows from \emph{(ii) $\rightarrow$ (i)}.
\end{proof}

To show that 2-orthonormality 
is in fact not a stronger notion than the usual orthonormality over $\FF_2$,
we will need an extension lemma:
\begin{lemma}
	Let $A$ be a binary matrix with linearly independent rows over $\FF_2$.
	Suppose an inhomogeneous linear equation $Ax = b$ has a solution $x = u^{\braket{\nu-1}}$ 
	mod $2^{\nu-1}$ where $\nu \ge 2$.
	(Here, the superscript $\braket{\nu-1}$ denotes the number $2^{\nu-1}$ modulo which the entries are defined.)
	Then, the same equation has a solution $x = u^{\braket{\nu}}$ mod $2^{\nu}$
	such that $u^{\braket{\nu}} = u^{\braket{\nu-1}} \mod 2^{\nu-1}$.
\label{lem:extension}
\end{lemma}
\begin{proof}
	The proof is by induction in the number of rows of $A$.
	When there is only one row in $A$, the scalar (i.e., a one-dimensional column vector)
	$A u^{\braket {\nu-1}} - b \mod 2^\nu$ is either 0 or $2^{\nu-1}$.
	We subtract this even number from the component of 
	$u^{\braket{\nu-1}}$ where $A$ has the first nonzero entry.
	The matrix $A$ must have such nonzero entry because the row is linearly independent over $\FF_2$.
	This proves the claim in the base case.
	
	To prove the induction step when $A$ has more than one row,
	consider Gauss elimination over $\ZZ/{2^{\nu}}\ZZ$ 
	to transform $A$ into $(1) \oplus A'$ up to permutations of columns.
	Such Gauss elimination is possible
	because the rows of $A$ are $\FF_2$-linearly independent, 
	and any odd number is invertible in $\ZZ/{2^{\nu}}\ZZ$.
	One finds a solution $u^{\braket{\nu}}|_{2,\ldots}$ for the $A'$ part 
	over $\ZZ/{2^{\nu}}\ZZ$ by the induction hypothesis,
	and then fixes the first component, if necessary, as in the base case.
\end{proof}

\begin{lemma}
Let $S$ and $L$ be rectangular matrices of $\FF_2$-linearly independent rows 
such that $SS^T = 0 \mod 2$, $S L^T = 0 \mod 2$, and $LL^T = I \mod 2$.
Then, there exists an odd integer coefficient vector $t$,
with respect to which $G = \begin{bmatrix} L \\ S \end{bmatrix}$ is $2$-orthogonal,
$L$ is $2$-orthonormal, and $S$ is $2$-null.
\label{lem:nu2basis}
\end{lemma}
\begin{proof}
The notion of $2$-orthogonality is defined over $\FF_2$, and therefore it is immediate that 
$G$, as a set of its rows, is $2$-orthogonal.
It remains to prove that there exists a coefficient vector $t$ such that
$\sum_i S_{ai} t_i = 0 \mod 4$ and $\sum_i L_{bi} t_i = 1 \mod 4$ for any $a,b$
(i.e., we need to check the $2$-norms). 
An odd integer solution $t$ to these equations are found by Lemma~\ref{lem:extension},
since these equations have solution $t = \vec 1 \mod 2$.
\end{proof}

\section{Level lifting}
\label{sec:levellifting}

\begin{theorem}\label{thm:levelLifting}
Let $G^{(\alpha)} = \begin{bmatrix} L^{(\alpha)} \\ S^{(\alpha)} \end{bmatrix}$ 
be a $(\nu,t^{(\alpha)})$-orthogonal binary matrix of $\nin^{(\alpha)}$ columns
for $\alpha = 1,2, \ldots, \ncheck$,
where $S^{(\alpha)}$ is $(\nu,t^{(\alpha)})$-null and consists of $\FF_2$-linearly independent rows,
$L^{(\alpha)}$ has $\nout$ rows,
and the set of all nonzero rows of $L^{(\alpha)}$ is $(\nu,t^{(\alpha)})$-orthonormal.
Assume that
\begin{align}
 \sum_k \sum_{i=1}^{\nin^{(\alpha)}} L^{(\alpha)}_{ki} t_i^{(\alpha)} &= \sum_i t_i^{(\alpha)} \mod 2^\nu. \label{eq:Ltt}
\end{align}

Let $M_{\alpha*}$ be a row binary vector defined as $M_{\alpha k} = 1$ if and only if $L^{(\alpha)}_{k*}$ is nonzero.
Then, the following matrix $\CF$ is $(\nu+1)$-orthogonal with respect to some coefficient vector $t$.
$\CF_\ell$ is $(\nu+1,t)$-orthonormal, and $\CF_\mathrm{out}$ and $\CF_\mathrm{in}$ are $(\nu+1,t)$-null.
Furthermore, $\sum_{r,i} (\CF_\ell)_{ri} t_i = \sum_i t_i$.
\begin{widetext}
\begin{align}
\begin{bmatrix}
\CF_\ell\\
\hline
\CF_\mathrm{out}\\
\hline
\CF_\mathrm{in}
\end{bmatrix}
= \left[
\begin{array}{c|c|c|c|c|c}
 I_\nout & (11)\otimes L^{(1)}          & (11) \otimes L^{(2)}         & (11) \otimes L^{(3)} & \cdots & (11) \otimes L^{(\ncheck)}\\
 \hline
 M_{1*} & (10)\otimes \vec 1           &                             &                     &          &                          \\
M_{2*} & M_{2*}((11)\otimes L^{(1)} )& (10)\otimes \vec 1           &                     &         &                          \\
M_{3*} & M_{3*}((11)\otimes L^{(1)} )& M_{3*}((11)\otimes L^{(2)} )& (10) \otimes \vec 1  &        &                          \\
\vdots  &          \vdots              &          \vdots              &        \vdots        &  \ddots&                           \\
M_{\ncheck *}&M_{\ncheck *}((11)\otimes L^{(1)})& M_{\ncheck *}((11)\otimes L^{(2)})& \cdots & \cdots & (10)\otimes \vec 1        \\
\hline
 & (11)\otimes S^{(1)} &  &  &  &  \\
 &                     &(11) \otimes S^{(2)} &   &  &  \\
 & &  &  \ddots & & \\
 &  & &  & (11) \otimes S^{(\ncheck-1)} & \\
 &  & &  & & (11) \otimes S^{(\ncheck)}
\end{array}
\right]
\end{align}
\end{widetext}
\end{theorem}

\begin{proof}
We will show that $t$ defined as (written as a row vector modulo $2^\nu$)
\begin{widetext}
\begin{align}
t=
\left[
\begin{array}{c|c|c|c|c|c}
	\vec 1_{\nout} & (-1~1)\otimes t^{(1)} & (-1~1)\otimes t^{(2)} & (-1~1)\otimes t^{(3)} & \cdots & (-1~1)\otimes t^{(\ncheck)}
\end{array}
\right]\label{eq:wt-struct}
\end{align}
\end{widetext}
satisfies the claim.
It is clear that $\sum_{r,i} (\CF_\ell)_{ri} t_i = \sum_i t_i$, even without modulo reduction.
It is also clear from the choice of $t$ that the $(\nu+1,t)$-norm of any row in $\CF_\ell$ is 1,
and that of any row in $\CF_\mathrm{in}$ is 0;
this does not depend on $t^{(\alpha)}$.
For the other conditions, we note that
\begin{align}
 \sum_{i=1}^{\nin^{(\alpha)}} L^{(\alpha)}_{ki} L^{(\alpha)}_{k'i} t^{(\alpha)}_i &= M_{\alpha k} \delta_{k k'} \mod 2^\nu. \label{eq:LLM}
\end{align}
If we set $k=k'$ in \eqref{eq:LLM}, sum over $k$, and use \eqref{eq:Ltt},
we see $\sum_i t_i^{(\alpha)}  = \sum_k M_{\alpha k} \mod 2^\nu$.
This implies that any row in $\CF_\mathrm{out}$ has $(\nu,t)$-norm zero.
To make the $(\nu+1)$-norm zero, we apply Lemma~\ref{lem:extension} to $\CF_\mathrm{out}$
since $\CF_\mathrm{out}$ is in a row echelon form (when read from the bottom right)
which ensures $\FF_2$-linear independence.
This may add $2^\nu$ to some components of $t$.
We are left with the $(\nu+1,t)$-orthogonality,
which is not affected by the modification to $t$ by Lemma~\ref{lem:extension}
since we will only need to evaluate sums modulo $2^\nu$.

We have to show that given $m$ rows with $2 \le m \le \nu+1$, 
the weighted sum of their overlap is zero modulo $2^{\nu+2-m}$.
Note that any part of the overlap that contains the tensor factor $(11)$
has no contribution to the weighted sum due to \eqref{eq:wt-struct}.
Let $\fa$ be a label set of chosen rows of $\CF_\ell$,
$\fb$ be that of $\CF_\mathrm{out}$,
and $\fc$ be that of $\CF_\mathrm{in}$.

{\it $\bullet$ $|\fa| \ge 2$ or $|\fc| \ge 1$}:

If $|\fb| = 0$, there is always a tensor factor $(11)$.
So, assume $|\fb| \ge 1$.
%For the moment, suppose that $\#(\fc) \ge 1$.
Except for the part with the tensor factor $(11)$,
we must show
\begin{align}
 2^{|\fa|+|\fb|+|\fc|-1} \sum_i L^{\star}_{\fa i} N^{\star}_{\fb' i} S^{\star}_{\fc i} t_i^{(\alpha)} = 0 \mod 2^{\nu+1}
\end{align}
where $L = L^{(\alpha)}$, $S=S^{(\alpha)}$, $N_{bi} = \sum_j M_{b j} L^{(\alpha)}_{ji} \mod 2$ 
for some $\alpha$ such that $\fb = \{\alpha\} \sqcup \fb'$ so $|\fb|=1+|\fb'|$.
Throughout the proof the range of the index $i$ is implicit,
but should be clear by the context.

Expanding $N^\star_{\fb' i}$ using \eqref{eq:gExpansion},
the left-hand side becomes a $\pm$-sum of terms
\begin{align}
 2^{|\fa|+|\fb'|+|\fc| + \sum_{j=1}^{|\fb'|}( |\mathfrak d^{(j)}|-1)}
 \sum_i
 L^{\star}_{\fa i} L^{\star}_{\mathfrak d^{(1)} i} \cdots L^{\star}_{\mathfrak d^{(|\fb'|)} i} S^{\star}_{\fc i} t_i^{(\alpha)} .
 \label{eq:lhs}
\end{align}
From the assumption of $\nu$-orthogonality of $G^{(\alpha)}$ (that includes $\nu$-nullity of $S^{(\alpha)}$),
we have
\begin{align}
&\text{ if } |\fa \cup_{j=1}^{|\fb'|} \mathfrak d^{(j)}| + |\fc| \ge 2 \text{ or } \fa=\fb'=\emptyset \nonumber \\
&\Longrightarrow \nonumber \\
&2^{|\fa \cup_j \mathfrak d^{(j)}| + |\fc|} 
 \sum_i 
 L^{\star}_{\fa i} L^{\star}_{\mathfrak d^{(1)} i} \cdots L^{\star}_{\mathfrak d^{(|\fb'|)} i} S^{\star}_{\fc i} t_i^{(\alpha)}
 = 0 \mod 2^{\nu+1} .
 \label{eq:from-orth}
\end{align}
The condition $|\fa \cup_j \mathfrak d^{(j)}| + |\fc| \ge 2$ or $\fa = \fb' = \emptyset$ is always true
in the present case:
It is clear if $|\fa| \ge 2$;
If $|\fa|=1$, then $|\fc| \ge 1$ so $|\fa| + |\fc| \ge 2$;
If $\fa = \emptyset$ but $|\fb| \ge 1$, then $|\fc| \ge 1$ and $|\mathfrak d^{(1)}| \ge 1$.
Comparing the exponents of 2 in \eqref{eq:lhs} and \eqref{eq:from-orth},
we see that \eqref{eq:lhs} is zero modulo $2^{\nu+1}$.

{\it $\bullet$ $\fa = \{a\}$ and $\fc = \emptyset$}:

Since we are choosing at least two rows, $|\fb| \ge 1$.
Let $\fb_1$ be the smallest row index of $\fb$.
Dropping the part with the tensor factor $(11)$,
we must show
\begin{align}
 2^{|\fb|} M^\star_{\fb a} - 2^{|\fb|} \sum_i L_{a i} N^{\star}_{\fb' i} t_i^{(\fb_1)} = 0 \mod 2^{\nu+1}
 \label{eq:target}
\end{align}
where $L = L^{(\fb_1)}$, $N_{bi} = \sum_j M_{b j} L^{(\fb_1)}_{ji} \mod 2$, and $\fb = \{\fb_1\} \sqcup \fb'$.
Expanding $N^\star_{\fb' i}$ using \eqref{eq:gExpansion},
the second term becomes a $\pm$-sum of terms
\begin{align}
 2^{|\fb'| + 1 + \sum_{j=1}^{|\fb'|}( |\mathfrak d^{(j)}|-1)}
 \sum_i
 L_{a i} L^{\star}_{\mathfrak d^{(1)} i} \cdots L^{\star}_{\mathfrak d^{(|\fb'|)} i} t_i^{(\fb_1)} .
 \label{eq:target1}
\end{align}
From the assumption of $\nu$-orthogonality of $G^{(\fb_1)}$, we have
\begin{align}
	&\text{ if } |\fa \cup_j \mathfrak d^{(j)}| \ge 2 	
\quad \Longrightarrow \quad 
	&2^{|\fa \cup_j \mathfrak d^{(j)}| } 
 \sum_i 
 L_{a i} L^{\star}_{\mathfrak d^{(1)} i} \cdots L^{\star}_{\mathfrak d^{(|\fb'|)} i} t_i^{(\fb_1)} 
 = 0 \mod 2^{\nu+1} .
\end{align}
Comparing the exponents of 2 of the last two expressions,
we see that only the terms with $\fa \cup_j \mathfrak d^{(j)} = \fa = \{a\}$ in \eqref{eq:target1} may survive.
(That is, $N$ can be substituted without higher order terms in \eqref{eq:gExpansion}.)
Hence, \eqref{eq:target} is equivalent to
\begin{widetext}
\begin{align}
& \begin{cases}
M^\star_{\fb a} - \sum_i L_{a i} t_i^{(\fb_1)} = 0 \mod 2^{\nu},	& \text{ if } \fb' = \emptyset, \\
2^{|\fb'|} M^\star_{\fb a} - 2^{|\fb'|} \sum_i L_{a i} \sum_j M^\star_{\fb' j } L_{ji} t_i^{(\fb_1)} = 0 \mod 2^{\nu}, &{\text{otherwise},}
\end{cases}
\nonumber\\
&\Longleftrightarrow 
 2^{|\fb'|} M^\star_{\fb a} - 2^{|\fb'|} \sum_i L_{a i} M^\star_{\fb' a} t_i^{(\fb_1)} = 0 \mod 2^{\nu} ,
\end{align}
\end{widetext}
where the last line follows from the $\nu$-orthogonality of $L$.
We know the last line is true since $\sum_i L^{(\fb_1)}_{a i} t_i^{(\fb_1)} = M_{\fb_1 a} \mod 2^\nu$
by \eqref{eq:LLM} and $\nu$-orthogonality of $G^{(\fb_1)}$.

{\it $\bullet$ $\fa = \fc = \emptyset$}:
By assumption, $|\fb| \ge 2$. Let $\fb_1$ be the smallest element of $\fb = \{\fb_1\} \sqcup \fb'$.
Except for the part with the tensor factor $(11)$,
we must show
\begin{align}
 2^{|\fb|-1} \sum_k  M^\star_{\fb k} - 2^{|\fb|-1} \sum_i N^{\star}_{\fb' i} t_i^{(\fb_1)} = 0 \mod 2^{\nu+1}
 \label{eq:target2}
\end{align}
where $L = L^{(\fb_1)}$, $S=S^{(\alpha)}$ and $N_{bi} = \sum_j M_{b j} L^{(\fb_1)}_{ji} \mod 2$.
Expanding $N^\star_{\fb' i}$ using \eqref{eq:gExpansion},
the second term becomes a $\pm$-sum of terms
\begin{align}
 2^{|\fb'| + \sum_{j=1}^{|\fb'|}( |\mathfrak d^{(j)}|-1)}
 \sum_i
 L^{\star}_{\mathfrak d^{(1)} i} \cdots L^{\star}_{\mathfrak d^{(|\fb'|)} i} t_i^{(\fb_1)} .
\end{align}
From the assumption of $\nu$-orthogonality of $G^{(\fb_1)}$, we have
\begin{align}
&\text{ if } |\cup_j \mathfrak d^{(j)}| \ge 2	
\quad \Longrightarrow \quad
2^{|\cup_j \mathfrak d^{(j)}| } 
 \sum_i 
 L^{\star}_{\mathfrak d^{(1)} i} \cdots L^{\star}_{\mathfrak d^{(|\fb'|)} i} t_i^{(\fb_1)}
 = 0 \mod 2^{\nu+1}.
\end{align}
Comparing the exponents of 2 of the last two expressions,
we see that only the terms with $\mathfrak d^{(j)} = \{d\}$ in \eqref{eq:target2} may survive.
Hence, \eqref{eq:target2} is equivalent to
\begin{align}
 2^{|\fb'|} \sum_k M^\star_{\fb k} - 2^{|\fb'|} \sum_i M^\star_{\fb' d} L_{di} t_i^{(\fb_1)} = 0 \mod 2^{\nu+1} ,
\end{align}
but we know this is satisfied since $\sum_i L_{d i} t_i^{(\fb_1)} = M_{\fb_1 d} \mod 2^\nu$
by \eqref{eq:LLM}.

We have shown the orthogonality condition (iii) of Lemma~\ref{lem:nu-orthogonality},
and completed the proof.
\end{proof}

\section{Transversal $T_\nu$ gates and Towers of codes}
\label{sec:tgates_codes}

In this section, we define a CSS code based on $\nu$-orthogonal matrices.
The resulting code will admit a transversal gate at the $\nu$-th level of Clifford hierarchy,
and can be used in a distillation protocol for the $(\nu+1)$-th level.
The results of previous sections will lead to a recursive construction of codes
of increasing levels of divisibility.

Let $\nu \ge 2$ be an integer, and define unitary matrices
\begin{align}
 T_\nu = \begin{pmatrix} 1 & 0 \\ 0 & \exp( 2\pi i/2^\nu ) \end{pmatrix}, \quad
 Z = \begin{pmatrix} 1 & 0 \\ 0 & -1 \end{pmatrix}, \quad
 X = \begin{pmatrix} 0 & 1 \\ 1 & 0 \end{pmatrix}.
\end{align}

\begin{definition}
Suppose we have a binary matrix $S$ with $n$ columns 
where the rows form a $(\nu,t)$-null set for some coefficient vector $t$,
and a $(\nu,t)$-orthogonal extension $G = \begin{bmatrix} L \\ S \end{bmatrix}$
where $L$ is $(\nu,t)$-orthonormal.
Let $\calS$ be the $\FF_2$-span of the rows of $S$,
and $\calG$ be the $\FF_2$-span of the rows of $G$.
We define a CSS code $\calQ[G,t]$, called a {\bf divisible (quantum) code at level $\nu$},
 on $n$ qubits by setting $\calS$ as $X$-stabilizers,
the rows of $L$ as the $X$- and $Z$-logical operators,
and $\calG^\perp$ (the orthogonal complement with respect to $2$-orthogonality, 
i.e., with respect to the usual dot product)
as the $Z$-stabilizers.
\end{definition}
Since each row of $L$ has odd weight, and is orthogonal (with respect to the usual dot product)
to any other row, the CSS code $\calQ[G,t]$ has $k$ logical qubits 
where $k$ is the number of rows in $L$.
For the code distance of $\calQ[G,t]$, we note the following.
\begin{lemma}
\[ 
d_X[G,S] := \min_{f \in \calG \setminus \mathcal S} |f| \ge  
 \min_{f \in \mathcal S^\perp \setminus \calG^\perp} |f| =: d_Z[G,S] .
\]
\label{lem:xdistance}
\end{lemma}
\begin{proof}
We claim that $\calG \setminus \calS \subseteq \calS^\perp \setminus \calG^\perp$,
implying that the minimum on the right-hand side is taken over a larger set.
Indeed, if $f \in \calG \setminus \calS$, 
then $f$ must have a nonzero dot product with some row of $L$
since $L$ is orthonormal.
This implies $f \notin \calG^\perp$. 
It is obvious that $\calG \subseteq \calS^\perp$
by the assumption on $G$.
\end{proof}
	
We now look at transversal gates.
As an extension of the calculation of \cite{BravyiHaah2012Magic} 
(see also \cite{LandahlCesare2013Tgate}),
we have
\begin{lemma}
On the code space of $\calQ[G,t]$, a transversal $T_\nu$ gate is the product of  
$\tilde T_\nu$ over all logical qubits, 
i.e., $\bigotimes_i T_\nu^{t_i} = \prod_{j=1}^k \tilde T_\nu$.
\label{lem:lopT}
\end{lemma}
\begin{proof}
For a bit string $(x_1,\ldots,x_k)$, the logical state in the computational basis
transforms as
\begin{widetext}
\begin{align}
 \left(\bigotimes_i T_\nu^{t_i} \right) \ket{\tilde x_1, \ldots, \tilde x_k} 
&= \frac{1}{\sqrt{|\calS|}} \sum_{s \in \calS} e^{2\pi i \sum_i f_i t_i /2^\nu} \ket{f = s + x_1 L_{1*} + \cdots + x_k L_{k*} \mod 2 }\nonumber\\
&= \frac{1}{\sqrt{|\calS|}}\sum_{s \in \calS} e^{2\pi i \sum_i x_i / 2^\nu} \ket{f = s + x_1 L_{1*} + \cdots + x_k L_{k*} \mod 2 }\nonumber\\
&= \left( \prod_{j=1}^k \tilde T_\nu \right) \ket{\tilde x_1, \ldots, \tilde x_k}.\label{eq:Taction}
\end{align}
\end{widetext}
The second equality is from Lemma~\ref{lem:nu-orthogonality}~(ii) and the orthonormality of $L$.
We see that the transversal $T_\nu$ with exponents matching the coefficient vector, 
implements logical $\tilde T_\nu$ on all logical qubits.
\end{proof}

Another logical action can be induced by a transversal gate;
namely, a fault-tolerant measurements in the basis of magic states 
$\ket{\psi_{\nu+1}} = (T_{\nu+1} X T_{\nu+1}^{-1}) \ket{\psi_{\nu+1}}$ and $Z\ket{\psi_{\nu+1}}$.
Before we present the induced action in Lemma~\ref{lem:ftmeasure} below,
we state a part of calculation in the following.
\begin{lemma}
$\bar X = X^{\otimes n}$ is a logical operator for $\calQ[G,t]$ where $n$ is the code length
if and only if $\vec 1 = \sum_{r} L_{r*} \mod 2$,
in which case the condition \eqref{eq:Ltt} in Theorem~\ref{thm:levelLifting} holds for $G$.
If $\vec 1 = \sum_{r} L_{r*}^{(\alpha)} \mod 2$ up to $S^{(\alpha)}$
for every $\alpha$ in Theorem~\ref{thm:levelLifting},
then $\vec 1 = \sum_r (\CF_\ell)_{r*} \mod 2$ up to $\CF_\mathrm{in,out}$.
\label{lem:1eqSumLog}
\end{lemma}
\begin{proof}
Since $Z$-stabilizers corresponds to $G^\perp$, it follows that 
$\vec 1$ is in the $\FF_2$-span of the rows of $G$.
Since $\vec 1 \cdot L_{r*} = 1 \mod 2$ for any $r$,
we see that $\vec 1 = \sum_r L_{r*} \mod 2$ up to $S$.
Using (ii) of Lemma~\ref{lem:nu-orthogonality}, 
we have $\norm{ \vec 1}_{\nu,t} = \sum_r \norm{L_{r*}}_{\nu,t} \mod 2^\nu$,
which is \eqref{eq:Ltt}.
	
All other claims are obvious.
\end{proof}

\begin{lemma}
Suppose $\bar X = X^{\otimes n}$ is a logical operator for $\calQ[G,t]$ 
where $n$ is the code length.
Then, $\bigotimes_j (T_{\nu+1}^{t_j} X T_{\nu+1}^{-t_j} )$ becomes
the product of logical operators $T_{\nu+1} X T_{\nu+1}^{-1}$ over all logical qubits 
up to an overall phase factor $\pm 1$.
\label{lem:ftmeasure}
\end{lemma}
\begin{proof}
Observe that $T_{\nu+1}^{m} X T_{\nu+1}^{-m} 
= \mathrm{diag}(e^{- 2\pi i m/ 2^{\nu+1}},e^{2\pi i m/ 2^{\nu+1}}) X 
= e^{-2\pi i m / 2^{\nu+1}} T_\nu^m X$
for any integer $m$.
Since Lemma~\ref{lem:1eqSumLog} says 
$\bar X$ is the product of the logical $X$ on all logical qubits,
it suffices to determine the action of 
$\bigotimes_j ( e^{-2\pi i t_j / 2^{\nu+1}} T_\nu^{t_j} ) = e^{-2\pi i \sum_j t_j / 2^{\nu+1}} \bigotimes_j T_\nu^{t_j}$
on the code space.
From \eqref{eq:Taction} we know $\bigotimes_j T_\nu^{t_j}$ is equal to $\prod_\ell \tilde T_{\nu}$ on the code space,
where $\ell$ ranges over all logical qubits.
Since $\bar X$ is logical, \eqref{eq:Ltt} holds by Lemma~\ref{lem:1eqSumLog},
and hence $e^{-2\pi i \sum_j t_j / 2^{\nu+1}}$ is equal to $\pm \prod_\ell e^{-2\pi i / 2^{\nu+1}} $.
\end{proof}

We are prepared to prove the main result.

\begin{theorem}
Let $s , \nu, \mu \ge 2 $ be integers such that $\nu \ge \mu$.
Let $\calQ[G,t]$ be a CSS code with parameters $[[n_\mu, k_\mu,2s+1]]$
where $G$ is $(\nu,t)$-orthogonal with respect to some odd coefficient vector $t$.
Assume that $X^{\otimes n_\mu}$ is a logical operator.
Then, there exists a CSS code with parameters $[[n_\nu, k_\nu, 2s+1]]$ that admits transversal $T_\nu$ as a logical operator
where
\begin{align}
 \frac{n_\nu}{k_\nu} = 
 (2s)^{\nu-\mu} \left(\frac{n_\mu}{k_\mu} + \frac{1}{2s-1}\right) - \frac{1}{2s-1}.
\end{align}
With $\nu=2$, the assumption is satisfied for every weakly self-dual CSS code of code distance $2s+1$.
\label{thm:tower}
\end{theorem}
\begin{proof}
We prove it by induction in $\nu$.
The base case $\nu = \mu$ is trivial,
since Lemma~\ref{lem:lopT} says $T_\nu$ can be implemented transversally.

For the induction step,
Lemma~\ref{lem:ftmeasure} implies there exists a circuit 
that measures the eigenvalue of 
$H_\mu^{\otimes k} = \bigotimes_{l=1}^k (T_{\mu+1} X T_{\mu+1}^{-1})$
with noisy $T_{\mu+1}$ gates on $\calQ[G,t]$.
Hence, there exists a distillation protocol for $T_{\mu+1}$ 
with the order of error reduction $2s+1$
that consumes $n_{\mu+1} = k_{\mu+1} + 2 n_\mu \ncheck$ 
noisy $T_{\mu+1}$ gates to produce $k_{\mu+1}$ magic states~\cite{HHPW2017magic}.
(See the review in Section~\ref{sec:review}.)
The distillation protocol is to check the eigenvalue of $H_\mu^{\otimes k}$
over $\ncheck$ subsets of input magic states,
so that no error pattern of weight $\le 2s$ is undetected.
It has been shown~\cite{HHPW2017magic} that
it is possible to specify the subsets of magic state to be tested
such that $k_{\mu+1} s = k_\mu \ncheck$.%
\footnote{
The choice of the subsets to be tested defines an ``outer code,''~\cite{HHPW2017magic}
which is $M$ in the matrix $\CF$ of Theorem~\ref{thm:completeCheckMatrix}.
The outer code $M$ is required to meet distance and sensitivity requirements
for the distillation protocol to have desired order of error reduction;
namely, $2|Me|+|e| \ge 2s+1$ for any nonzero vector $e$.
The condition $k_{\mu+1} s = k_\mu \ncheck$ 
means that the outer code must not be too redundant,
but the existence of such an outer code is 
given by a graph theoretic construction~\cite{HHPW2017magic}.
For $s \ge 3$, one can instead use a $s$-dimensional ``grid'' outer code~\cite{enum}.
}

Invoking Theorem~\ref{thm:completeCheckMatrix},
we find a complete check matrix $G'$ for the entire protocol.
The complete check matrix satisfies \ref{eq:Ltt} due to Lemma~\ref{lem:1eqSumLog},
and therefore is $(\mu+1,t')$-orthogonal by Theorem~\ref{thm:levelLifting}
for some $t'$.
We thus obtain a code $\calQ[G',t']$ on $n_{\mu+1}$ qubits, encoding $k_{\mu+1}$ logical qubits.
By Lemma~\ref{lem:xdistance}
the code distance of $\calQ[G',t']$ 
is the order of error reduction of the distillation protocol,
which is $2s+1$.
To complete the induction, we finally need to show that 
the transversal $X$ is a logical operator,
but we know this by the last statement of Theorem~\ref{thm:levelLifting}
and Lemma~\ref{lem:1eqSumLog}.

Any weakly self-dual CSS code $[[n_2,k_2,2s+1]]$ 
has the transversal $X$ as a logical operator.
Since the code distance $2s+1$ is odd,
such weakly self-dual CSS code must be normal.
%(The first arXiv version of \cite{HHPW2017magic} referred normal codes as ``odd'' codes.)
By Lemma~\ref{lem:nu2basis} we find a coefficient vector $t$
with respect to which the code $[[n_2,k_2,2s+1]]$ admit transversal logical operator 
$\bigotimes_i (T^{t_i}_2 X)$.
\end{proof}

\section{Examples}
\label{sec:examples}

We briefly survey existing codes, and explain how they do or do not fit into our tower.

\subsection{With good ground codes}

We can start the tower with a good family of normal weakly self-dual CSS codes 
that is indexed by the code distance $d$ and has parameters 
$[[n=O(d), k=\Omega(d), d]]$~\cite{CalderbankShor1996Good,HHPW2017magic}.
It has been shown~\cite{HHPW2017magic}
that we can turn an instance of this code family
into a $H^{\otimes \le k}$-measurement gadget,
and by checking the eigenvalue of products of Hadamard
over $O(d)$ random subsets of $O(k)$ noisy magic states,
we can implement a distillation circuit with the order of error reduction $d$.
Proceding analogously to the proof of Theorem~\ref{thm:tower},
we obtain for each $\nu$ a family
of codes with parameters $[[O(d^{\nu-1}), \Omega(d), d]]$
that admit a transversal logical gate $T_\nu$ at the $\nu$-th level of Clifford hierarchy.

\subsection{Quantum Reed-Muller codes of code distance 3}

Though the statement of Theorem~\ref{thm:tower} does not include the case $s=1$,
here we claim that our tower starting with the Steane $[[7,1,3]]$ code
gives the family of quantum Reed-Muller codes of distance $3$,
that are used in \cite{BravyiKitaev2005Magic,LandahlCesare2013Tgate}.
In particular, the code at level $\nu=3$ is the code $[[15,1,3]]$ of \cite{BravyiKitaev2005Magic},
and our construction explicitly proves the equivalence between the distillation protocols 
by Knill~\cite{Knill2004a} and Bravyi-Kitaev~\cite{BravyiKitaev2005Magic}.

The quantum Reed-Muller codes of \cite{BravyiKitaev2005Magic,LandahlCesare2013Tgate} have
the $X$-stabilizer spaces as shortened first order Reed-Muller codes.
The generating matrix of the first order Reed-Muller code on $2^m$ bits can be recursively defined as
\begin{align}
G_{1,1} = \begin{bmatrix} 1 & 1 \\ 0 & 1 \end{bmatrix},\quad
G_{1,m > 1	} = \begin{bmatrix} G_{1,m-1} & G_{1,m-1} \\ 0 & \vec 1 \end{bmatrix}
\end{align}
and the generating matrix $S_m$ of the shortened code 
is obtained by removing the first row, which is $\vec 1$, and the first column from $G_{1,m}$.
(For a general definition of shortened and punctured classical codes, see \cite{MacWilliamsSloane}.)
The dual of $S_m$ contains $\vec 1$ of odd weight, and has the minimum distance 3.
Hence, if we declare that the rows of $S_m$ define the $X$-stabilizers,
$\vec 1$ defines both $X$ and $Z$ logical operators,
and $(\mathrm{row~span~} S_m + \{0,\vec 1\} )^\perp$ defines the $Z$-stabilizers,
then we arrive at a $[[2^m -1, 1, 3]]$ code.
Since every code word of $S_m$ has weight a multiple of $2^{m-1}$,
the quantum code $[[2^m -1, 1, 3]]$ admits a transversal 
$\mathrm{diag}(1,\exp(2\pi i / 2^{m-1})$-gate~\cite{LandahlCesare2013Tgate}.
Spelling out the generating matrix for $X$-stabilizers, we see
\begin{align}
G_{1,m+1}=\begin{bmatrix}
	1 & \vec 1 & 1 & \vec 1 \\
	0 & S_m    & 0 & S_m \\
	0 & 0      & 1 & \vec 1
\end{bmatrix} 
\xrightarrow{\text{shorten}}
S_{m+1} = \begin{bmatrix}
S_m    & 0 & S_m \\
0      & 1 & \vec 1
\end{bmatrix}.
\end{align}
This is the same as the $X$-stabilizer matrix
from our level lifting when $\nout =1$, up to permutation of columns and rows:
\begin{align}
	\begin{bmatrix}
		L \\
	\hline
		S
	\end{bmatrix} \longrightarrow 
	\begin{bmatrix}
		1 & L      & L \\
	\hline
		1 & \vec 1 & 0 \\
		0 & S      & S
	\end{bmatrix}.
\end{align}
Since $L = \vec 1$, the claim is proved.
Higher order Reed-Muller codes give nonexamples to our tower; see below.

\subsection{Triorthogonal codes}

There exists a distillation protocol using a CSS code 
where $\bigotimes T$ over the physical qubits
is not necessarily a logical operator.
This is based on the idea that $\bigotimes T$ over the physical qubits
need not be a logical operator by itself, 
but only has to be logical up to Clifford corrections.
A sufficient condition for this is so-called {\it triorthogonality}~\cite{BravyiHaah2012Magic}
imposed on a generating matrix of $X$-logical operators and $X$-stabilizers.
Consider a binary matrix $G$ consisting of even weight rows $G_0$ and odd weight rows $G_1$.
The triorthogonality requires that 
every pair and every triple of distinct rows of $G$ have even overlaps.
Given such a triorthogonal matrix, one defines a CSS code by declaring that
the rows of $G_0$ correspond to $X$-stabilizers, and $G^\perp$ to $Z$-stabilizers,
and each row of $G_1$ corresponds to both $X$- and $Z$-logical operators.
For example, one can consider a classical triply even code, and puncture it so that
$X$-stabilizers have weight a multiple of 8 and $X$-logical operators have odd weight,
to obtain a triorthogonal codes~\cite{BravyiKitaev2005Magic,rtrio}.
In this case, the Clifford correction is trivial.
More generally, the divisibility condition at level 3 that is relaxed to hold modulo 2, is precisely the triorthogonality.
Hence, any of our generalized divisible codes at level 3 is an example of triorthogonal codes.
In summary, we have inclusion relations among classes of codes, as depicted in Figure~\ref{fig:codesAt3}.
\begin{figure}[t]
	\includegraphics[width=.48\textwidth]{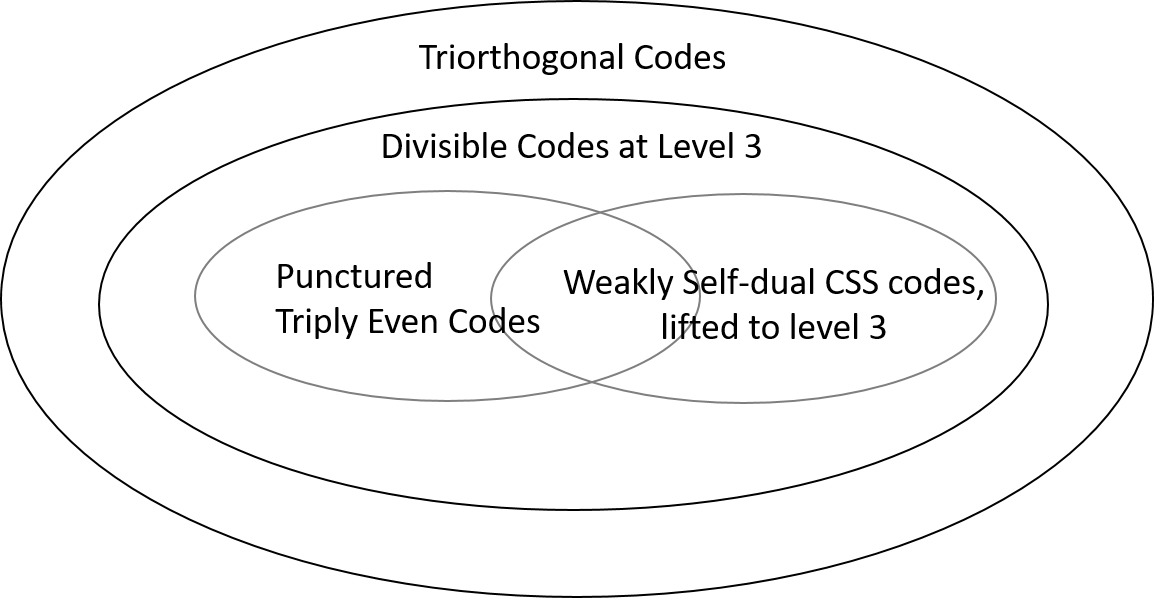}
\caption{Relations of codes that admit a transversal gate at Clifford hierarchy level 3.
The inner two sets are proper subsets of the set of all divisible codes at level 3,
but we do not know whether the set of all divisible codes at level 3 coincides with the set of all triorthogonal codes.
}
\label{fig:codesAt3}
\end{figure}
The inclusion of the two smallest sets are strict as we will show shortly,
but we do not know whether the outer-most inclusion is strict.

Now we show the family of triorthogonal codes $[[3k+8,k,2]]$ with $k$ even in \cite{BravyiHaah2012Magic}
is actually a divisible code at level 3, level-lifted from the $H$-code of Jones~\cite{Jones2012}.
The relation was loosely expected in \cite{Jones2012} by comparing the leading term
in the error probability when the codes are used in distillation protocols.
Since this triorthogonal code is not a punctured triply even code, 
the set of all weakly self-dual CSS codes, lifted to level 3,
is a proper subset of the set of all divisible codes at level 3.

The ``$H$-code~\cite{Jones2012}'' is a normal weakly self-dual CSS code defined for any even $k \ge 2$ by
\begin{align}
\begin{bmatrix}
	0 & 0 & 1 & 0 & \cdots & 0 & 1 & 1 \\
	0 & 0 & 0 & 1 & \cdots & 0 & 1 & 1 \\
	  &   &   &   & \ddots &   &   &   \\
	0 & 0 & 0 & 0 & \cdots & 1 & 1 & 1 \\
	\hline
	1 & 0 & 1 & 1 & \cdots & 1 & 1 & 0 \\
	0 & 1 & 1 & 1 & \cdots & 1 & 0 & 1 \\
\end{bmatrix}	
=
\begin{bmatrix}
	0_{k \times 2} & I_k & 1_{k \times 2} \\
	\hline
	I_2            & 1_{2 \times k} & I_2
\end{bmatrix}
\end{align}
where the rows of the upper submatrix represent the logical operators,
and the bottom the stabilizers.
The distillation protocol is such that $k$ qubits are checked once 
by the $H$-measurement routine implemented by the $H$-code.
Going through our level-lifting, we obtain
\begin{align}
\begin{bmatrix}
	I_k            & (11) \otimes 	(0_{k \times 2} ~~ I_k ~~ 1_{k \times 2}) \\
\hline
	1_{1 \times k} & (01) \otimes 1_{1 \times (k+4)} \\
	0_{2 \times k} & (11) \otimes (I_2      ~~1_{2 \times k} ~~ I_2)
\end{bmatrix}
=
\begin{bmatrix}
	I_k            & 0_{k \times 4} & I_k            & I_k           & 1_{k \times 4} \\
\hline
	1_{1 \times k} & 0011           & 0_{1 \times k} & 1_{1 \times k} & 0011 \\
0_{1 \times k} & 1010      & 1_{1 \times k} & 1_{1 \times k} & 1010\\
0_{1 \times k} & 0101      & 1_{1 \times k} & 1_{1 \times k} & 0101     
\end{bmatrix}
\end{align}
Permuting the columns and applying a row operation, we have
\begin{align}
\Longrightarrow
\begin{bmatrix}
     0_{k \times 4}& 1_{k \times 4} & I_k            & I_k            &		I_k            \\
\hline
	 0011  & 0011                   & 1_{1 \times k} & 0_{1 \times k} & 1_{1 \times k}  \\
	 1010  & 1010                   & 0_{1 \times k} & 1_{1 \times k} & 1_{1 \times k} \\
	 0101  & 0101                   & 0_{1 \times k} & 1_{1 \times k} & 1_{1 \times k}      
	\end{bmatrix}
\cong
	\begin{bmatrix}
	0_{k \times 4}& 1_{k \times 4} & I_k            & I_k            &		I_k            \\
\hline
	0011  & 0011                   & 1_{1 \times k} & 0_{1 \times k} & 1_{1 \times k}  \\
	1010  & 1010                   & 0_{1 \times k} & 1_{1 \times k} & 1_{1 \times k} \\
	1111  & 1111                   & 0_{1 \times k} & 0_{1 \times k} & 0_{1 \times k}      
\end{bmatrix},
\end{align}
which is the triorthogonal matrix of \cite{BravyiHaah2012Magic}.

\subsection{Doubling transformation}

The pipelining of Ref.~\cite{HHPW2017magic} using two-dimensional color codes~\cite{Bombin_2006}
that encodes one logical qubit, 
yields a family of codes with parameters $[[O(d^3),1,d]]$ that admits transversal logical $T_3$.
This family matches that of \cite{BravyiCross2015,Jones2016,JochymOConnorBartlett2016};
in fact, the ``doubling transformation'' of \cite{BravyiCross2015,Jones2016,JochymOConnorBartlett2016}
is a special case of our tower from level 2 to 3,
where the outer code's parity check matrix is $M=\begin{bmatrix} 1 \end{bmatrix}$,
which defines the zero code on one bit.
Examples of divisible codes at level 3 that can be obtained this way 
include
$[[15,1,3]]$~\cite{BravyiKitaev2005Magic} and $[[49,1,5]]$~\cite{BravyiHaah2012Magic}.

\subsection{Nonexamples}

Punctured Reed-Muller codes $RM(r,m)$ of $Z$-distance larger than 4 with $m > \nu r$
do not fit into our tower,
but belong to the set of divisible codes.
The reason is as follows.
On one hand, our level-lifted divisible codes are degenerate;
there always exists a $Z$-stabilizer of weight 4,
but the $Z$-distance can be larger than 4.
On the other hand, a punctured Reed-Muller code's $Z$-stabilizer 
can have minimal weight $2^{r+1}$, where $r$ is the order of the Reed-Muller code~\cite{sd}.

Another class of potential nonexamples includes
three or higher dimensional color codes~\cite{Bombin_2006,KubicaBeverland2015}.
The divisor and the scaling of code parameters of the codes from our tower
match (up to constants) with those of the color codes, 
but it is unclear whether the stabilizers of the code from our tower can be put on a lattice 
with the geometric locality of the generators obeyed.

\section{Discussion}
\label{sec:discussion}

We have unified two large classes of distillation protocols
by showing that every instance of the distillation protocols of Ref.~\cite{Knill2004a,HHPW2017magic}
based on Clifford measurements using normal codes
corresponds to a single generalized divisible code that admits a transversal non-Clifford logical gate.
The correspondence is exact, so one can analyze the error of the protocol in Ref.~\cite{HHPW2017magic}
via the resulting generalized divisible code.
Note that this unification does not cover all known protocols,
especially those that do not always accept the output even in the limit of noiseless non-Clifford 
gates~\cite{BravyiKitaev2005Magic,Reichardt2005,HowardDawkins2015};
if a protocol admits a representation in terms of a single block code with a transversal non-Clifford gate,
then the protocol must always accept the output upon noiseless non-Clifford inputs.

As a distillation protocol, the protocol of Ref.~\cite{HHPW2017magic}
can be thought of as ``spacetime'' protocol,
since not all noisy $T$-gates are applied initially;
in fact, noisy $T$-gates that are consumed initially are only a small fraction of all the noisy $T$ gates.
In contrast, the protocol using the corresponding divisible code as derived here
can be thought of as a ``space-only'' protocol,
since all noisy $T$-gates are used at the very beginning.
Direct comparison of the two approaches is not appropriate, 
without further information about qubit and Clifford gate implementations.
The spacetime protocol has a relatively small space footprint at the price of a large $T$-depth,
whereas the space-only protocol requires a large space overhead, but has a $T$-depth of 1.
The true depth including all Clifford gates of the space-only protocol is probably smaller
than that of the spacetime protocol.
This may not be obvious since the decoding routine (the inverse of encoding)
generally becomes more complicated as the code length increases.
However, a large portion of the stabilizers of the space-only protocol is block-diagonal 
(see Theorem~\ref{thm:completeCheckMatrix}),
which can be encoded/decoded in parallel.

In \cite{MeierEastinKnill2012Magic-state,HHPW2017magic},
there are distillation protocols for $T$ gates using hyperbolic inner codes.
Then, a natural question is whether a corresponding divisible code at level 3 exists.
We do not know the answer for this question, which is rather roughly posed,
but a similar method presented above will not work.
At least, there does not exist any triorthogonal code with parameters $[[10,2,2]]$,
which is conceivable by the protocol of \cite{MeierEastinKnill2012Magic-state};
the smallest triorthogonal code with code distance 2 is on 14 qubits~\cite[Sec.~5]{CampbellHoward2017b}.
The reason is that an error pattern on an inner code does not lead to a deterministic syndrome.
In an $H$-measurement routine using a hyperbolic code,
the relevant errors are $Y$ errors, but any logical $Y$ is not a tensor product of physical $Y$ operators,
and hence undetected $Y$ errors propagate to the logical qubits as some combination of $X$ and $Z$ errors.
These propagated errors are beyond the error model where only $Y$ errors are considered.
One could restore the error model by inserting a twirling channel $\rho \mapsto \frac 1 2 \rho + \frac 1 2 H \rho H$
after each $H$-measurement routine,
but then a deterministic error becomes a probabilistic error, 
and our complete check matrix will not be able to capture this behavior.

For a similar reason, it should be possible to find a divisible code at level 3,
starting with distillation protocols for $CS$ and $CCZ$ gates~\cite{rtrio}.
Another natural question is whether one can construct towers of codes for qudits
by relating distillation protocols based on Clifford measurements~\cite[App.~D]{HHPW2017magic}
and those based on divisible codes~\cite{Campbell2012,Howard2012,Campbell2014}.

\begin{acknowledgments}
I thank M.~B. Hastings and D. Poulin for encouraging discussions.
\end{acknowledgments}	

\bibliography{divisible-ref}
\end{document}